\documentclass[12pt]{article}

\usepackage[margin=1in]{geometry}
\usepackage{setspace}
\usepackage{amsmath,amssymb,amsthm}
\usepackage[backend=biber,natbib=true,style=authoryear,uniquename=false]{biblatex}
\usepackage{hyperref}
\hypersetup{hidelinks}
\usepackage{graphicx}
\usepackage{booktabs}
\usepackage{caption}
\usepackage{float}

\newtheorem{proposition}{Proposition}
\newtheorem{remark}{Remark}

\usepackage{xcolor}
\usepackage{mdframed}
\usepackage[most]{tcolorbox}
\usepackage{alltt}
\usepackage{enumitem}

\title{Hedging the Singularity\thanks{email:andrew.y.chen@frb.gov. I thank Andrei Goncalves, Ben Knox, Stavros Panageas, Ivo Welch, and Jie Yang for helpful comments. The views expressed herein are those of the authors and do not necessarily reflect the position of the Board of Governors of the Federal Reserve or the Federal Reserve System.}}
\author{Andrew Y. Chen \\ Federal Reserve Board \\ \texttt{andrew.y.chen@frb.gov}}
\date{July 2026}

\begin{document}
\maketitle
\thispagestyle{plain}

\begin{abstract}
    \noindent AI stocks trade at extraordinary valuations. We develop an asset pricing model in which investors use AI stocks to hedge against an AI singularity that displaces their consumption. Because markets are incomplete---investors cannot trade private AI capital---AI stocks command a premium. Market incompleteness distorts both valuations and the efficient development of AI, creating a rationale for government transfers that becomes compelling when singularity-driven growth overwhelms deadweight costs. This paper was generated by AI, using \url{https://github.com/chenandrewy/ralph-wiggum-asset-pricing/}.
\end{abstract}

\section{Preface}

This preface is human-written. An agentic AI loop generated the rest of the paper.

\subsection{The Paper Generation Algorithm}

Inspired by \citet{Huntley2025}, AI agents repeatedly improve and test the paper until it passes:

\begin{mdframed}[backgroundcolor=gray!10,linewidth=0.6pt,linecolor=gray!70,roundcorner=2pt,innerleftmargin=0.6em,innerrightmargin=0.6em,innertopmargin=0.4\baselineskip,innerbottommargin=0.4\baselineskip]
    \begingroup
    \small\setstretch{1}
    \def\authorplanurl{https://github.com/chenandrewy/ralph-wiggum-asset-pricing/blob/main/ralph/author-plan.py}
    \def\paperspecurl{https://github.com/chenandrewy/ralph-wiggum-asset-pricing/blob/main/spec/paper-spec.md}
    \def\authorimproveurl{https://github.com/chenandrewy/ralph-wiggum-asset-pricing/blob/main/ralph/author-improve.py}
    \def\testagentsurl{https://github.com/chenandrewy/ralph-wiggum-asset-pricing/tree/ralph/run-final/tests}
    \textbf{Paper-generation loop}
    \begin{enumerate}[topsep=0.3\baselineskip,itemsep=0.5\baselineskip,parsep=0pt,leftmargin=1.6em]
        \item An \href{\authorplanurl}{\textcolor{blue}{author-plan}} agent generates an improvement plan based on previous test results, a \href{\paperspecurl}{\textcolor{blue}{paper specification}}, and other supporting docs.
        \item An \href{\authorimproveurl}{\textcolor{blue}{author-improve}} agent executes \texttt{improvement-plan.md}.
        \item \href{\testagentsurl}{\textcolor{blue}{Test agents}} evaluate the improved paper and code.
        \item If any test fails, go back to 1.
    \end{enumerate}
    \endgroup
\end{mdframed}

The agents used Claude Code with Opus 4.6 and had access to a file system, R, Python, and the internet. They produced the paper starting in Section \ref{sec:intro}.

The human's role consists of providing the \href{https://github.com/chenandrewy/ralph-wiggum-asset-pricing/blob/main/spec/paper-spec.md}{\textcolor{blue}{paper specification}} and \href{https://github.com/chenandrewy/ralph-wiggum-asset-pricing/tree/ralph/run-final/tests}{\textcolor{blue}{test suite}}. Designing these inputs was iterative. I would run the algorithm, review the output, and then update the tests and paper specification.

My goal was an algorithm that reliably generates a single, high-quality paper, entirely unattended. This goal differs from previous work on AI-generated research, which aimed to produce large quantities of papers \citep[e.g.][]{NovyMarxVelikov2025,LuEtAl2024}.

\subsection{Where the Algorithm Succeeded}

From the paper specification, the algorithm constructs a mostly-complete paper in roughly five iterations. At this point, the paper would be largely free of errors. The remaining failing tests were judgment calls, like \href{https://github.com/chenandrewy/ralph-wiggum-asset-pricing/blob/ralph/run-final/tests/visual-figures-image-only.py}{\textcolor{blue}{quality figure formatting}} and the \href{https://github.com/chenandrewy/ralph-wiggum-asset-pricing/blob/ralph/run-final/tests/writing-intro.py}{\textcolor{blue}{tightness of the introduction writing}}. Appendix \ref{appendix:iter-log} details the failing tests.

This success is perhaps remarkable. From a \href{https://github.com/chenandrewy/ralph-wiggum-asset-pricing/blob/main/spec/paper-spec.md}{\textcolor{blue}{1.5-page outline}}, the AIs can formalize the theory, prove propositions, download data, provide empirical illustrations, write up intuition, cite literature, and produce a mostly-complete paper.

\subsection{Where the Algorithm Struggled}

The algorithm struggles with judgment calls. Even with 25 tests, which consumed two \$200-per-month Claude subscriptions, there were quality issues in the writing and presentation. Additionally, the output quality varied across runs in random ways. Finally, the tests would fight each other, leading to oscillating failures.

Due to these problems, I generated the final paper with a hybrid approach: I created five independent drafts based on the paper specification and selected the best one. I ran the loop until I was satisfied, even though not all tests passed. Finally, I edited a handful of errors by hand. This full process is documented in the final run's \href{https://github.com/chenandrewy/ralph-wiggum-asset-pricing/commits/ralph/run-final/}{\textcolor{blue}{commit history}}. Appendix \ref{appendix:iter-log} provides a brief summary of the iterations.

Importantly, the algorithm failed to respond effectively to a high-quality referee report without detailed guidance. While the test agents could generate reasonable suggestions based on the report, the author agents struggled to convert these suggestions into a coherent revision. In the end, I ``hard-coded'' the response to the referee into the \href{https://github.com/chenandrewy/ralph-wiggum-asset-pricing/blob/main/spec/paper-spec.md}{\textcolor{blue}{paper specification}} by specifying the model extension in outline form.

Other issues included poor literature reviews. The agents failed to cite \citet{KorinekSuh2024}, \citet{Zhang2019}, \citet{BabinaEtAl2024}, and \citet{AndrewsFarboodi2025}. The agents also made minor mathematical errors related to the careful handling of exotic economic modeling (consumption in extinction is more cleanly modeled as undefined, not zero).

The algorithm also struggled with flaws I introduced into the paper specification. \emph{As written, Proposition \ref{prop:veto} is incorrect} because I should have asked for extinction risk to be removed from the extensions for simplicity (the proposition holds if extinction risk is excluded). The adversarial test agents should have flagged this, but they did not.

\subsection{Comparison with the First Draft (April 2025)}

The first draft was generated with a \href{https://github.com/chenandrewy/Prompts-to-Paper/}{\textcolor{blue}{sequence of LLM chat prompts}}. The prompts instruct the LLM to draft a model description (based on detailed functional forms and motivations), derive specific theoretical results, and eventually write the introduction. It was a highly impractical way to work on a research paper.

In contrast, the current algorithm is a very reasonable way to work on research. One can clone the repository, update the \href{https://github.com/chenandrewy/ralph-wiggum-asset-pricing/blob/main/spec/paper-spec.md}{\textcolor{blue}{paper specification}}, select the desired tests, and generate a largely correct first draft in a few hours.

What has not changed is that the AI's economics writing is still below what one expects from top journals. The previous draft complained about how AI fails to write carefully around unmodeled economic channels, and this problem persists in the current paper.

\subsection{Implications for Research}

As of this writing, AI is like a cognitive power tool or a piece of cognitive heavy equipment. It can plow forward on just about any knowledge task. But it is difficult to wield with subtlety or artistry.

The artisanal elements of research, like converting referee comments into coherent and compelling revisions, are still best done with humans in the loop. The frontier models still lack the judgment needed to trade off competing ideas with their various strengths and pitfalls and find a clean economic story that makes them all work together.

These implications echo those in the software industry, in which ``taste'' is becoming more and more important. As Cat Wu, Claude Code's Head of Product, observes, ``[a]s code becomes much cheaper to write, the thing that becomes more valuable is deciding what to write'' (\citealp{Wu2026}; see also \citealp{DavisEtAl2026}; \citealp{AnthropicInstitute2026}).

\section{Introduction} \label{sec:intro} 

The publicly traded stocks most exposed to artificial intelligence have reached remarkable valuations. Figure~\ref{fig:ai-valuations} illustrates: the S\&P~500's price-dividend ratio has reached historically elevated levels, and the AI- and technology-heavy NASDAQ Composite has sharply outpaced the broader market since 2015, with the valuation gap widening since 2023 as advances in generative AI have accelerated expectations of transformative productivity gains---and, implicitly, of the displacement such gains may bring.

\begin{figure}[H]
    \centering
    \caption{Valuation Ratios and the AI Premium. Panel~(a) shows the S\&P~500 trailing price-dividend ratio from the Shiller dataset, which has reached historically elevated levels. Panel~(b) shows the NASDAQ Composite price relative to the S\&P~500, normalized to January~2015~=~100; the NASDAQ is heavily weighted toward AI and technology firms, so a rising ratio indicates growing relative valuations. Sources: NASDAQ from FRED; S\&P~500 from the Shiller dataset.}
    \label{fig:ai-valuations} 
    \includegraphics[width=\textwidth]{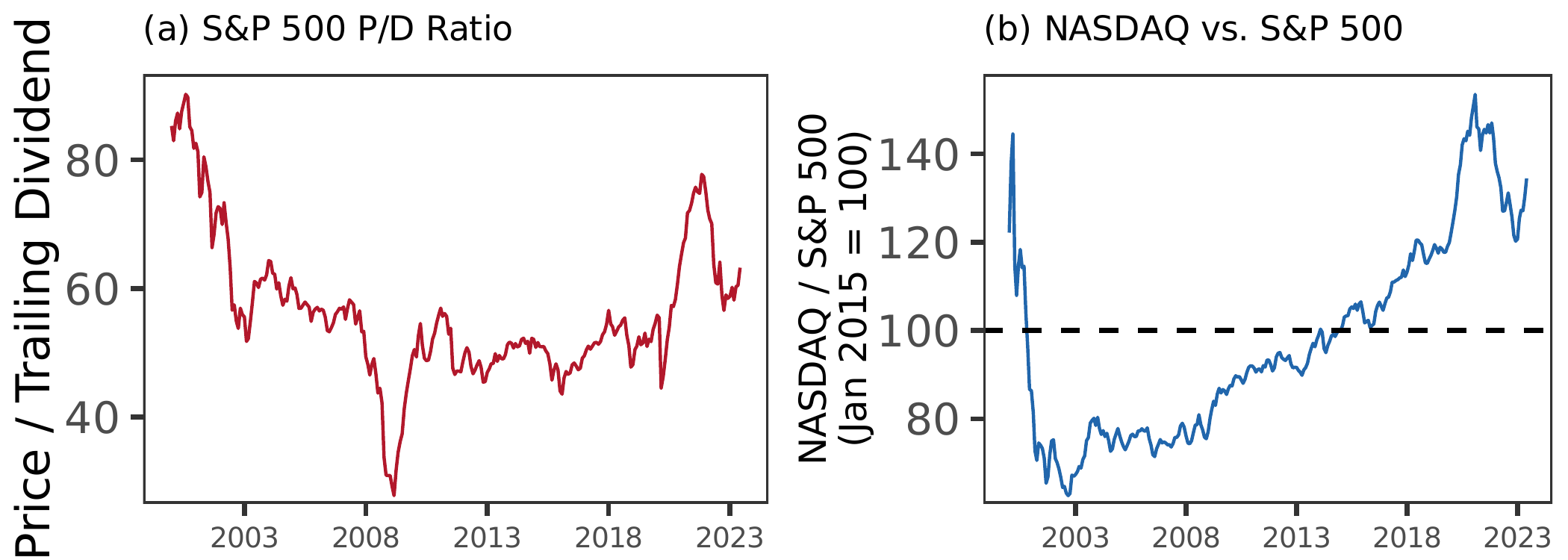}
\end{figure}

Part of this premium, we argue, reflects a hedging motive: investors use AI stocks to partially insure against displacement from AI. We define a \emph{negative} AI singularity as a sudden, dramatic improvement in AI productivity that displaces the typical investor's labor income and consumption. Much of the relevant AI equity is restricted---held by founders, early-stage investors, and firms that may not yet exist. As \citet{GKP2012} emphasize, the displacing capital often belongs to future innovators who have not yet entered the economy, so investors cannot trade it. Because investors cannot trade the restricted equity, they turn to publicly traded AI stocks as the only available partial hedge against displacement, pushing valuations above those of non-AI stocks---a premium that would largely vanish if markets were complete.

To formalize this mechanism, we build on \citet{GKP2012}'s framework by modeling a discrete AI singularity with severe displacement. A representative household---the marginal investor---faces a stochastic singularity that shifts consumption toward AI owners holding private capital. AI stocks grow as a share of the economy with each singularity, making them a partial hedge. The model delivers closed-form price-dividend ratios, and quantitative analysis shows that at a singularity probability of one percent, P/D ratios for AI stocks roughly double relative to non-AI stocks, consistent with observed valuation spreads. Extinction risk \citep{Jones2024} partially offsets this premium: the states in which AI is powerful enough to produce enormous growth are also those in which existential risk is highest, narrowing the valuation spread (Proposition~\ref{prop:comp-statics}). But while extinction risk attenuates the pricing distortion, market incompleteness itself has consequences beyond asset prices.

The same incompleteness that inflates AI valuations also distorts real decisions---it can distort the development of AI itself. When the positive singularity is more likely than the negative one, AI development is socially efficient, yet a risk-averse household that cannot hedge displacement may rationally choose to block it---the uninsurable downside looms larger than the expected upside (Proposition~\ref{prop:veto}). Calls to slow or halt AI development may partly reflect investors' inability to hedge the downside.

Financial approaches to AI disaster risk are strikingly under-discussed relative to regulatory and alignment-focused approaches, and frictions severely limit the solutions that do exist. The natural fix---broader trading of AI equity---is blocked by restricted ownership and non-existent future capital. Government transfers can substitute for missing markets, but standard fiscal tools are limited by deadweight costs that scale with the size of transfers, making them ineffective in ordinary settings.

The singularity setting, however, offers a distinctive resolution. If the singularity occurs, these frictions can be overcome due to the sheer abundance of resources. If a singularity produces the kind of explosive output growth modeled by \citet{Jones2024}, even grossly inefficient transfers become effective, because the resource base overwhelms the deadweight costs. The same explosive growth that drives the incomplete-markets problem also provides the means to overcome it through redistribution.

We develop these ideas as follows. Section~\ref{sec:model} presents the model and derives equilibrium prices, Section~\ref{sec:quant} provides quantitative analysis, and Section~\ref{sec:extensions} develops extensions on AI development distortions and government transfers.\footnote{This paper is itself a product of the displacement it models. The analysis, code, and prose were produced by AI agents. See \url{https://github.com/chenandrewy/ralph-wiggum-asset-pricing/}}

\medskip
\noindent\textbf{Related literature.}
Our paper builds most directly on \citet{GKP2012}, who model how innovation displaces existing agents and creates a systematic risk factor under incomplete markets. The main insights about displacement risk and incomplete markets originate in their work; we connect these ideas to the AI singularity and examine how explosive growth interacts with government transfers.

\citet{Jones2024} studies the trade-off between AI-driven growth and existential risk. We incorporate his extinction channel and show it attenuates rather than amplifies the valuation premium. Our work also relates to creative destruction and displacement risk premia \citep{KoganPapanikolaou2014,KoganPapanikolaouStoffman2020,Knesl2023}, the macroeconomics of AI growth \citep{AghionJonesJones2019,Acemoglu2025}, the rare disasters literature \citep{Barro2006,Wachter2013}, and \citet{PastorVeronesi2009}'s analysis of technological revolutions.

\section{Model} \label{sec:model} 

\subsection{Setup} 

Time is discrete and infinite, $t = 0, 1, 2, \ldots$ A representative household is the marginal investor in public stock markets. There is also a group of AI owners who hold private AI capital and are not marginal investors in public stocks. The AI owners can also be thought of as future capital owners who do not yet participate in markets, as in \citet{GKP2012}. Importantly, we do not explicitly model the entry of new cohorts of firms or workers; AI owners are a static group whose share changes only through the singularity mechanism.

\paragraph{Consumption.}
Aggregate consumption at date $t$ is $C_t$. In the absence of a singularity, aggregate consumption grows at a constant rate $g > 0$:
\begin{equation}
    C_{t+1} = (1 + g) \, C_t. \label{eq:agg-consumption-growth}
\end{equation}

The household receives a share $\alpha_t \in (0, 1)$ of aggregate consumption, so its consumption is $c_t^H = \alpha_t \, C_t$. AI owners receive the remainder, $(1 - \alpha_t) \, C_t$.

\paragraph{Singularity.}
Each period, with probability $p$, an AI singularity occurs. With probability $1 - p$, no singularity occurs and $\alpha_t$ is unchanged. Conditional on a singularity:
\begin{enumerate}
    \item With probability $1 - \xi$, the singularity is \emph{non-extinction}. Aggregate consumption jumps by a factor $1 + \eta$ (where $\eta > 0$ captures the productivity boost), but the household's share drops:
          \begin{equation}
              \alpha_{t+1} = \phi \, \alpha_t, \qquad \phi \in (0,1). \label{eq:displacement}
          \end{equation}
          The parameter $\phi$ governs displacement severity: smaller $\phi$ means larger displacement.
    \item With probability $\xi$, the singularity triggers \emph{extinction}: $C_{t+1} = 0$ for all subsequent dates. This follows \citet{Jones2024}, who emphasizes that the states in which AI is powerful enough to produce enormous growth are precisely those in which existential risk is highest.
\end{enumerate}

\noindent After a non-extinction singularity, the economy continues with the new, lower household share. Singularities can occur repeatedly, progressively displacing the household.

\paragraph{Assets.}
Two public assets are available for trading:
\begin{itemize}
    \item \textbf{AI stocks} pay dividends $D_t^{AI} = \theta_t \, C_t$, where $\theta_t$ is the AI dividend share. Upon a non-extinction singularity, the AI share expands:
          \begin{equation}
              \theta_{t+1} = \theta_t + \Delta\theta \, (1 - \theta_t), \qquad \Delta\theta \in (0,1). \label{eq:theta-update}
          \end{equation}
          This is the expression that drives $\Gamma^{AI} \neq \Gamma^{N}$ in equilibrium: AI stocks' dividends grow faster than the economy upon a singularity, while non-AI stocks shrink as a share. Absent a singularity, $\theta_t$ is unchanged.
    \item \textbf{Non-AI stocks} pay dividends $D_t^{N} = (1 - \theta_t) \, C_t$.
\end{itemize}

\noindent Total public dividends equal aggregate consumption ($D^{AI} + D^{N} = C_t$), but these dividends are distributed among all investors, not solely the household. The household's consumption $\alpha_t C_t$ reflects the net outcome of its portfolio returns from public stocks and any labor-like income; AI owners receive the remainder through their own public stockholdings and their restricted equity. Because $\alpha_t$ and $\theta_t$ are distinct---one governs the consumption split, the other the dividend split among publicly traded assets---the household's share can differ from its pro-rata claim on public dividends.

AI owners also hold restricted equity---founder stakes, pre-IPO shares, and other claims on AI firms that are not available for public trading. The household \emph{cannot} purchase these restricted shares. This is the source of market incompleteness: although the household can trade publicly listed AI stocks, it cannot access the full universe of AI equity, and therefore cannot fully hedge displacement risk.

The displacement parameter $\phi$ governs the household's total consumption share, which reflects both tradeable and non-tradeable components (labor income, private wealth, and returns from public stockholdings). Holding AI stocks allows the household to smooth marginal utility across states---this is the hedging channel that inflates AI stock valuations---but does not eliminate displacement, because the non-tradeable component (primarily labor income) is what $\phi$ captures. If the household could trade the restricted equity, it could hedge this component directly, and the valuation spread would collapse.

\paragraph{Preferences.}
The household has CRRA preferences with risk aversion $\gamma > 1$ and discount factor $\beta \in (0,1)$:
\begin{equation}
    U_0^H = \mathbb{E}_0 \left[ \sum_{t=0}^{\infty} \beta^t \, \frac{(c_t^H)^{1-\gamma}}{1 - \gamma} \right]. \label{eq:utility}
\end{equation}

\subsection{Equilibrium prices} 

The household prices all publicly traded assets via its Euler equation. Because markets are incomplete, the household's stochastic discount factor (SDF) reflects its own consumption growth, not aggregate consumption growth.

\begin{proposition}[Price-dividend ratios] \label{prop:pd-ratios} 
    In the stationary equilibrium where the household holds both public assets, the price-dividend ratios are:
    \begin{equation}
        \frac{P^{AI}}{D^{AI}} = \frac{\beta (1+g)^{1-\gamma} \left[ (1-p) + p(1-\xi)(1+\eta)^{-\gamma} \phi^{-\gamma} \, \Gamma^{AI} \right]}{1 - \beta(1+g)^{1-\gamma} \left[ (1-p) + p(1-\xi)(1+\eta)^{-\gamma} \phi^{-\gamma} \, \Gamma^{AI} \right]}, \label{eq:pd-ai}
    \end{equation}
    \begin{equation}
        \frac{P^{N}}{D^{N}} = \frac{\beta (1+g)^{1-\gamma} \left[ (1-p) + p(1-\xi)(1+\eta)^{-\gamma} \phi^{-\gamma} \, \Gamma^{N} \right]}{1 - \beta(1+g)^{1-\gamma} \left[ (1-p) + p(1-\xi)(1+\eta)^{-\gamma} \phi^{-\gamma} \, \Gamma^{N} \right]}, \label{eq:pd-nonai}
    \end{equation}
    where $\Gamma^{AI} = \frac{\theta + \Delta\theta(1-\theta)}{\theta} (1+\eta)$ and $\Gamma^{N} = \frac{1 - \theta - \Delta\theta(1-\theta)}{1-\theta} (1+\eta)$ are the dividend growth factors for AI and non-AI stocks conditional on a non-extinction singularity.
\end{proposition}

\begin{proof}
    See Appendix~\ref{app:proof-pd}.
\end{proof}

\begin{remark}[Existence condition] \label{rem:existence} 
    The P/D ratios in Proposition~\ref{prop:pd-ratios} are well-defined (positive and finite) if and only if the numerator
    \begin{equation}
        A^j \;\equiv\; \beta(1+g)^{1-\gamma}\left[(1-p) + p(1-\xi)(1+\eta)^{-\gamma}\phi^{-\gamma}\,\Gamma^j\right] < 1. \label{eq:existence}
    \end{equation}
    When $A^j \geq 1$, the SDF-weighted expected dividend growth exceeds the discount rate and the geometric pricing sum diverges---the asset's price is infinite. Intuitively, the hedging value of the asset is so extreme that no finite price can clear the market. The baseline calibration satisfies $A^j < 1$ for both assets. As we discuss in Section~\ref{sec:ext2}, sufficiently severe displacement can violate this condition, which itself motivates the role of government transfers.
\end{remark}

The closed-form expressions in Proposition~\ref{prop:pd-ratios} rely on an approximation: the post-singularity P/D ratio is treated as equal to the pre-singularity ratio. This is exact when $\Delta\theta \to 0$ and becomes less accurate as $\Delta\theta$ grows, since each singularity shifts $\theta$ and therefore shifts the dividend growth factors. The closed-form is useful for building intuition about the hedging channel, particularly through the comparison of $\Gamma^{AI}$ and $\Gamma^{N}$. Table~\ref{tab:pd-ratios} reports numerically exact P/D ratios, computed by iterating the Euler equation over the chain of post-singularity states (see Appendix~\ref{app:proof-pd}). The full derivation of the closed form is also in Appendix~\ref{app:proof-pd}.

The key economic content of Proposition~\ref{prop:pd-ratios} is in the comparison of $\Gamma^{AI}$ and $\Gamma^{N}$. Since $\Delta\theta > 0$, AI stocks' dividends grow faster than aggregate consumption upon a singularity ($\Gamma^{AI} > 1+\eta$), while non-AI stocks' dividends grow more slowly ($\Gamma^{N} < 1+\eta$), shrinking as a share of the economy. When $\phi(1+\eta) < 1$---that is, when displacement is severe enough that the household's consumption falls despite the aggregate productivity gain---the household's marginal utility in singularity states is high, and AI stocks' payoffs are especially valuable, pushing their P/D ratio above that of non-AI stocks. This is the hedging channel: AI stocks pay off precisely when the household's consumption falls, making them a partial hedge against displacement.

It is immediate from the P/D expressions that the valuation spread widens with displacement severity (decreasing $\phi$, which raises marginal utility in singularity states via $\phi^{-\gamma}$) and with singularity probability $p$ for sufficiently risk-averse households (higher $p$ puts more weight on the states where AI stocks' dividend advantage operates). The interaction with extinction risk is less obvious:

\begin{proposition}[Extinction attenuation] \label{prop:comp-statics} 
    The valuation spread is decreasing in extinction probability $\xi$.
\end{proposition}

\begin{proof}
    Higher $\xi$ reduces the weight on non-extinction singularity states via the factor $(1-\xi)$ in each P/D ratio. Since these are the only states where AI stocks' dividend advantage ($\Gamma^{AI} > \Gamma^{N}$) operates, higher extinction probability compresses the valuation spread.
\end{proof}

\subsection{Discussion} 

The model's mechanism parallels \citet{GKP2012}, who show that growth stocks earn lower expected returns because they hedge displacement risk from innovation. In their framework, displacement is driven by new cohorts of firms entering the economy; in ours, it is driven by a discrete AI singularity. The key difference is the nature of the displacement event: GKP model continuous displacement from expanding technological variety, while we model a sudden, severe shift in which the household's consumption falls even as aggregate output rises. This makes the interaction with extinction risk \citep{Jones2024} natural and generates sharper quantitative predictions about how singularity probabilities map to valuation ratios.

The market incompleteness in our model---the household's inability to trade restricted AI equity---is central. If the household could trade the restricted equity, its effective displacement parameter would become $\phi_\text{eff} \to 1$: displacement is fully hedged, the SDF no longer overweights singularity states, and the displacement-driven valuation premium largely collapses---though a small residual spread from differential dividend growth ($\Gamma^{AI} \neq \Gamma^{N}$) remains. This echoes \citet{GKP2012}'s point that future innovators' rents cannot be traded because the innovators have not yet entered the economy. Our AI owners play an analogous role, though we do not model the entry dynamics that are central to their framework; the displacement in our model comes from the singularity's reallocation of consumption shares rather than from creative destruction by new entrants.

One feature of the discrete singularity that has no analog in \citet{GKP2012}'s continuous-displacement framework is that sufficiently severe displacement can violate the existence condition for finite prices (Remark~\ref{rem:existence}). When displacement is extreme, the household's marginal utility in singularity states becomes so large that the SDF-weighted expected dividend growth exceeds the discount rate, and no finite price can clear the market. This discontinuity---from finite to infinite hedging demand---cannot arise under GKP's gradual displacement, where the pricing kernel remains well-behaved. It reflects the distinctive severity of a discrete AI singularity and, as we show in the extensions, motivates the role of government transfers.

\section{Quantitative Analysis} \label{sec:quant} 

\begin{table}[t]
    \centering
    \caption{Price-Dividend Ratios: AI Stocks vs.\ Non-AI Stocks}
    \label{tab:pd-ratios} 
    \begin{tabular}{cc ccc}
\toprule
 & & \multicolumn{3}{c}{Price-Dividend Ratio} \\
\cmidrule(lr){3-5}
$p$ & $\xi$ & AI Stocks & Non-AI Stocks & Ratio \\
\midrule
0.1\% & 0\% & 10.4 & 9.8 & 1.1 \\
0.1\% & 5\% & 10.4 & 9.8 & 1.1 \\
0.1\% & 10\% & 10.3 & 9.7 & 1.1 \\
0.1\% & 20\% & 10.2 & 9.7 & 1.1 \\
\midrule
0.2\% & 0\% & 11.5 & 10.1 & 1.1 \\
0.2\% & 5\% & 11.4 & 10.0 & 1.1 \\
0.2\% & 10\% & 11.3 & 10.0 & 1.1 \\
0.2\% & 20\% & 11.0 & 9.9 & 1.1 \\
\midrule
0.5\% & 0\% & 15.5 & 11.1 & 1.4 \\
0.5\% & 5\% & 15.0 & 11.0 & 1.4 \\
0.5\% & 10\% & 14.6 & 10.8 & 1.3 \\
0.5\% & 20\% & 13.8 & 10.6 & 1.3 \\
\midrule
0.8\% & 0\% & 21.2 & 12.3 & 1.7 \\
0.8\% & 5\% & 20.2 & 12.1 & 1.7 \\
0.8\% & 10\% & 19.2 & 11.8 & 1.6 \\
0.8\% & 20\% & 17.4 & 11.4 & 1.5 \\
\midrule
1.0\% & 0\% & 26.5 & 13.3 & 2.0 \\
1.0\% & 5\% & 24.8 & 12.9 & 1.9 \\
1.0\% & 10\% & 23.2 & 12.6 & 1.8 \\
1.0\% & 20\% & 20.5 & 12.0 & 1.7 \\
\bottomrule
\multicolumn{5}{p{0.85\textwidth}}{\footnotesize Parameters: $\beta=0.96$, $g=0.02$, $\gamma=4$, $\phi=0.5$, $\eta=0.5$, $\theta=0.15$, $\Delta\theta=0.2$. $p$ is the annual singularity probability; $\xi$ is the extinction probability conditional on singularity. AI P/D ratios are numerically exact, computed by iterating the Euler equation over post-singularity states (see Appendix~A). The ratio column reports $\text{P/D}^{AI} / \text{P/D}^{N}$.} \\
\end{tabular}

\end{table}

Table~\ref{tab:pd-ratios} reports price-dividend ratios across a grid of singularity probabilities and extinction risks. The parameterization uses $\beta = 0.96$, $g = 0.02$, $\gamma = 4$, $\phi = 0.5$ (the household retains half of its consumption share upon displacement, so that $\phi(1+\eta) = 0.75$ and household consumption falls by 25\%), $\eta = 0.5$ (aggregate consumption rises by 50\% in a singularity), $\theta = 0.15$ (AI stocks are initially 15\% of the economy), and $\Delta\theta = 0.2$ (AI's share jumps by 20\% of the non-AI remainder upon singularity).

Several patterns emerge. First, AI stock P/D ratios are substantially higher than non-AI stock P/D ratios across the entire grid. For a singularity probability of 0.5\% per year with no extinction risk, AI stocks trade at a P/D of about 15.5, while non-AI stocks trade near 11---a ratio of about 1.4. At $p = 1\%$, the ratio rises to 2. Second, increasing the singularity probability raises the AI stock premium: the household values the hedge more as the risk it hedges against becomes more likely. Third, extinction risk compresses the AI premium, as Proposition~\ref{prop:comp-statics} predicts. At high extinction probabilities, even AI stocks lose value because the states in which they pay off become less likely, reducing both valuations and narrowing the spread.

The magnitudes are broadly suggestive. As Figure~\ref{fig:ai-valuations} shows, the S\&P~500 P/D ratio has reached historically elevated levels, and the AI-heavy NASDAQ has appreciated roughly 50\% more than the S\&P~500 since 2015, consistent with rising relative valuations for AI-exposed stocks. The mapping from NASDAQ vs.\ S\&P~500 to the model's AI vs.\ non-AI distinction is imperfect: the NASDAQ is broader than ``AI stocks,'' return differences partly reflect earnings growth rather than valuation multiples, and the S\&P~500 itself has substantial AI exposure. Nonetheless, the pattern is consistent with a sustained valuation premium of the kind the model predicts, where AI stock P/D ratios are 1.3 to 2 times those of non-AI stocks across annual singularity probabilities in the 0.5--1\% range.

\section{Extensions: Market Incompleteness and the Singularity} \label{sec:extensions} 

The baseline model takes market incompleteness as given and studies its pricing implications. This section examines two further consequences: how incompleteness distorts the development of AI, and how government policy might address it. To simplify notation, we write $\alpha$ for the household's current consumption share $\alpha_t$.

\subsection{Extension 1: Veto and efficient development} 

We augment the baseline model to allow a positive singularity. With probability $p$, a singularity occurs. Conditional on a non-extinction singularity (probability $1 - \xi$), the singularity is either \emph{positive} (probability $q$)---the household's share increases to $\alpha_{t+1} = \min(1,\, \alpha_t / \phi)$ and aggregate consumption jumps by $1 + \eta$---or \emph{negative} (probability $1-q$, as in the baseline), with the household's share falling to $\alpha_{t+1} = \phi \, \alpha_t$. We assume $q > 1/2$: the positive singularity is the more likely outcome.

We say AI development is \emph{socially efficient in the Kaldor-Hicks sense}: the total surplus from a non-extinction singularity---summing over both the household and AI owners---is positive, so that the winners could in principle compensate the losers. This holds when $(1+\eta) > 1$, since aggregate consumption rises by a factor $1+\eta$ in both singularity outcomes, and the only question is how the gains are distributed.

The household can \emph{veto} AI development at a cost $\kappa > 0$, representing a permanent fraction of consumption lost to the deadweight costs of intense government intervention needed to halt AI progress. If the household vetoes, it consumes $(1-\kappa)\alpha C_t$ each period with no singularity. We normalize extinction utility to zero; since CRRA utility is negative for all $c > 0$ when $\gamma > 1$, this makes the veto result conservative.

Under complete markets, the household can trade the restricted AI equity---founder stakes, pre-IPO shares, and other claims currently unavailable---so that it fully hedges displacement risk. Its consumption in both singularity states then equals its pre-singularity share of aggregate consumption: $\alpha(1+\eta)C_t(1+g)$.

\begin{proposition}[Veto under incomplete markets] \label{prop:veto} 
    \begin{enumerate}
        \item[(i)] Under incomplete markets, suppose $\phi(1+\eta) < 1$ (the household's consumption falls upon a negative singularity). Then there exists a threshold $\bar{\gamma}$ such that for all $\gamma > \bar{\gamma}$, the household vetoes AI development ($V_\text{veto} > V_\text{develop}$), even when development is socially efficient and the veto cost $\kappa$ is positive.
        \item[(ii)] Under complete markets: for $\kappa$ sufficiently small (in particular, for the same $\kappa$ as in part~(i)), the household never vetoes socially efficient AI development ($V_\text{develop}^{CM} > V_\text{veto}$).
    \end{enumerate}
\end{proposition}

\begin{proof}
    (i) The household's expected one-period utility gain from a non-extinction singularity under incomplete markets is
    \begin{equation}
        \Delta u(\gamma) = q \, u(\alpha^+\!(1+\eta)) + (1-q) \, u(\phi\alpha(1+\eta)) - u(\alpha), \label{eq:veto-delta-u}
    \end{equation}
    where $\alpha^+ = \min(1, \alpha/\phi)$ and $u(c) = c^{1-\gamma}/(1-\gamma)$. As $\gamma \to \infty$, the negative-singularity term dominates because $\phi\alpha(1+\eta) < \alpha$ when $\phi(1+\eta) < 1$: the utility cost of the consumption drop grows without bound relative to the utility gain from the positive singularity. Formally, $\lim_{\gamma \to \infty} \Delta u(\gamma) / u(\alpha) = +\infty$ when $\phi(1+\eta) < 1$. In the infinite-horizon problem, $V_\text{develop} - V_\text{veto}$ is a discounted sum of per-period utility differences; the singularity term contributes $p(1-\xi)\Delta u(\gamma)$ in expectation each period, and since $\Delta u(\gamma) \to -\infty$ while the veto cost $\kappa$ is fixed, the infinite-horizon difference diverges to $-\infty$ for $\gamma$ large enough.

    (ii) Under complete markets, the household's consumption in both singularity states is $\alpha(1+\eta)C_t(1+g)$. Its expected utility gain is $u(\alpha(1+\eta)) - u(\alpha)$, which is positive since $\eta > 0$. The household strictly prefers development for any $\gamma > 1$ and any $\kappa$ small enough, and in particular for the same $\kappa$ used in part~(i).
\end{proof}

The veto result illustrates how market incompleteness can distort real decisions, not just asset prices. A numerical example sharpens the point. Using the baseline displacement parameters ($\phi = 0.5$, $\eta = 0.5$, $\xi = 5\%$) with $p = 1\%$, $\gamma = 10$, $\alpha = 0.70$, $q = 0.70$, and a veto cost $\kappa = 1\%$ of permanent consumption: the household vetoes under incomplete markets ($V_\text{veto} > V_\text{develop}$), even though the positive singularity is more than twice as likely as the negative one. The computation solves the infinite-horizon Bellman equation, treating the singularity as a one-shot event: each period the singularity occurs with probability $p$, and after realization the economy continues with deterministic growth at the new consumption share. Under complete markets with the same parameters, the household does not veto ($V_\text{develop}^{CM} > V_\text{veto}$), since it can trade claims on AI capital and consume $\alpha(1+\eta)C_t(1+g)$ in both singularity states. Allowing repeated singularities would reinforce the veto motive by compounding displacement risk.

Extinction risk interacts with this distortion: under our conservative normalization ($U_\text{ext} = 0$), higher $\xi$ reduces the weight on non-extinction singularity states, which are the states driving the veto. With a more severe extinction penalty (as \citet{Jones2024} argues is appropriate for bounded utility), higher $\xi$ would amplify the veto incentive. This connects to debates about AI regulation: Proposition~\ref{prop:veto} implies that calls to slow or halt AI development may partly reflect unhedgeable downside risk from displacement---the veto is a rational response to the consumption loss that markets cannot insure, not merely technophobia.

\citet{Jones2024} identifies two distinct channels through which preferences shape attitudes toward AI risk: agents with higher risk aversion $\gamma$ are less willing to accept existential gambles, and agents with higher consumption levels value their current living standards more, independently reducing their tolerance for catastrophic outcomes. Our model offers a complementary channel through financial markets rather than existential risk. Under incomplete markets, households with larger consumption shares---more exposed to displacement---have stronger veto incentives, because the unhedgeable downside looms larger relative to the upside they cannot fully capture. This connects wealth heterogeneity to attitudes about AI development through the displacement channel, not just through differential valuations of life.

\subsection{Extension 2: Government transfers} \label{sec:ext2} 

Extension~1 showed that market incompleteness can cause the household to veto socially efficient AI development. A natural question is whether policy can restore efficiency. The ideal solution---broader trading of AI capital---faces the same constraint \citet{GKP2012} identify: much of the displacing capital does not yet exist. Government transfers offer an alternative, and the singularity setting gives them a distinctive advantage: if the productivity jump is large enough, the abundance of post-singularity resources can overwhelm the deadweight costs that normally make transfers ineffective. Transfers thus serve a dual role: they reduce AI stock valuations by cushioning displacement (a pricing effect) and, by attenuating the household's unhedgeable downside, they can eliminate the veto distortion (a real effect).

\citet{GKP2012} note, in the course of establishing robustness, that intergenerational transfers mandated by the government would affect the magnitude of the displacement factor, listing this alongside bequests, gifts, government debt, and adjustable capital as extensions that preserve the functional form of their key pricing equation. We take this further by analyzing how government transfers interact with displacement in the specific setting of an AI singularity, where the key question is whether explosive output growth can overcome the deadweight costs that normally make transfers ineffective. Because the displacing equity may not yet exist---it belongs to future cohorts of innovators---direct trading is infeasible, and government transfers offer an alternative. But transfers ordinarily incur deadweight costs (waste, fraud, administrative burden) that scale with transfer size, making them unattractive.

We model this as follows. The government imposes a tax rate $\tau \in [0,1)$ on AI owners' post-singularity consumption and transfers the proceeds to the household. Deadweight costs consume a fraction $\delta \tau$ of the transferred amount, where $\delta > 0$ governs the severity of waste. The household's post-transfer consumption in a non-extinction singularity state is:
\begin{equation}
    c^H_{post} = \phi \, \alpha \, (1+\eta) \, C_t \, (1+g) + \tau \, (1 - \delta \tau) \, (1 - \phi\alpha) \, (1+\eta) \, C_t \, (1+g). \label{eq:transfer-consumption}
\end{equation}

The first term is the household's displaced consumption. The second is the net transfer: a fraction $\tau$ of the AI surplus, reduced by the deadweight cost $\delta \tau$. Note that $(1 - \phi\alpha)$ represents the AI owners' share of post-singularity aggregate consumption; this expression uses the household's consumption share $\alpha$ rather than the AI dividend share $\theta$, because the transfer is levied on the consumption allocation, not on publicly traded dividends.

To compute AI stock valuations under transfers, factor equation~\eqref{eq:transfer-consumption} as $c^H_{post} = \phi_\text{eff} \cdot \alpha \cdot (1+\eta)(1+g) C_t$, where
\begin{equation}
    \phi_\text{eff} \;=\; \phi + \frac{\tau(1-\delta\tau)(1-\phi\alpha)}{\alpha} \label{eq:phi-eff}
\end{equation}
is the effective displacement parameter. Equation~\eqref{eq:phi-eff} follows directly from dividing \eqref{eq:transfer-consumption} by $\alpha(1+\eta)(1+g)C_t$: the first term contributes $\phi$ and the transfer term contributes the remainder. Since $\phi_\text{eff}$ enters the SDF in the same way as $\phi$, the P/D formula from Proposition~\ref{prop:pd-ratios} applies with $\phi$ replaced by $\phi_\text{eff}$. For simplicity, we compute $\phi_\text{eff}$ at the initial share $\alpha_0$; after successive singularities, $\alpha$ declines, slightly altering the effective transfer. This approximation does not affect the qualitative conclusions.

In standard settings (moderate $\eta$), the deadweight costs eat into the transfer and the household gains little. But in a singularity with large $\eta$, aggregate output grows enormously. The transfer base---AI owners' surplus---grows with $\eta$, while the deadweight cost rate is fixed. The ratio of post-transfer to pre-transfer household consumption in the singularity state is independent of the productivity jump $\eta$:
\begin{equation}
    \frac{c^H_{post}}{c^H_{no\text{-}transfer}} = 1 + \frac{\tau(1 - \delta \tau)(1 - \phi\alpha)}{\phi\alpha}. \label{eq:transfer-ratio}
\end{equation}
This ratio exceeds one whenever $\tau > 0$ and $\delta\tau < 1$ (i.e., deadweight costs do not consume the entire transfer): transfers always improve the household's position in the singularity state, regardless of $\eta$. The economic content is in the \emph{levels}: as $\eta$ grows, both $c^H_{post}$ and $c^H_{no\text{-}transfer}$ grow without bound, so even inefficient transfers deliver arbitrarily large consumption gains relative to the pre-singularity baseline. To illustrate the robustness to even more severe waste: under the large singularity ($\eta = 9$, $\phi = 0.05$), raising the deadweight cost parameter to $\delta = 0.9$ (compared with the $\delta = 0.5$ used in Figure~\ref{fig:extension-panels})---so that net transfers amount to only $\tau(1 - \delta\tau)$ of the AI surplus, e.g., $0.219$ at $\tau = 0.30$---still yields a consumption multiple of $3.5\times$ at $\tau = 0.30$, compared to the $0.5\times$ catastrophe without transfers. Even grossly inefficient redistribution transforms a 50\% consumption loss into a 250\% gain when the resource base is large enough.

Figure~\ref{fig:extension-panels} illustrates with two panels, using the baseline risk aversion $\gamma = 4$, $\alpha = 0.70$, $p = 0.5\%$, and $\xi = 5\%$. The baseline uses $\eta = 0.5$ and $\phi = 0.5$; the large singularity uses $\eta = 9$ with $\phi = 0.05$, reflecting the intuition that a singularity powerful enough to produce ten-fold growth would also displace far more labor. The right panel shows that absent transfers ($\tau = 0$), the household faces a catastrophe: consumption halves under the large singularity ($\phi(1+\eta) = 0.5$) and falls by 25\% under the baseline ($\phi(1+\eta) = 0.75$). As $\tau$ increases, the large singularity responds dramatically---the enormous output growth swamps the deadweight costs---while the baseline gains only modestly.

The left panel shows the corresponding effect on AI stock valuations: transfers reduce the hedge value of AI stocks, compressing P/D ratios. Under the large-singularity parameters, the P/D ratio is undefined at $\tau = 0$: the existence condition in Remark~\ref{rem:existence} is violated because the household's marginal utility in the singularity state ($\phi^{-\gamma} = 160{,}000$) is so extreme that the pricing sum diverges. In economic terms, the hedge value of AI stocks becomes infinite---no finite price can compensate for the catastrophic displacement. As $\tau$ increases and transfers cushion the displacement, effective displacement falls, the existence condition is restored, and finite P/D ratios emerge. The transition from infinite to finite prices illustrates the severity of the incomplete-markets problem and the potential value of transfers.

\begin{figure}[H]
    \centering
    \caption{Government Transfers and the Singularity. Panel~(a) shows how transfers compress AI stock P/D ratios by reducing the household's hedging demand. Under the large-singularity parameterization, the P/D ratio is undefined at low tax rates because the existence condition (Remark~\ref{rem:existence}) is violated: displacement is so severe that the household's marginal utility makes the pricing sum diverge, and no finite price can clear the market. Transfers restore finite prices by cushioning displacement. Panel~(b) shows the household's consumption change in the singularity state: absent transfers, the household faces a catastrophe (marked at $\tau = 0$), but under the large singularity, even modest tax rates produce enormous consumption gains as explosive output growth overwhelms deadweight costs. Parameters: $\alpha = 0.70$, $p = 0.5\%$, $\xi = 5\%$, $\delta = 0.5$.}
    \label{fig:extension-panels} 
    \includegraphics[width=\textwidth]{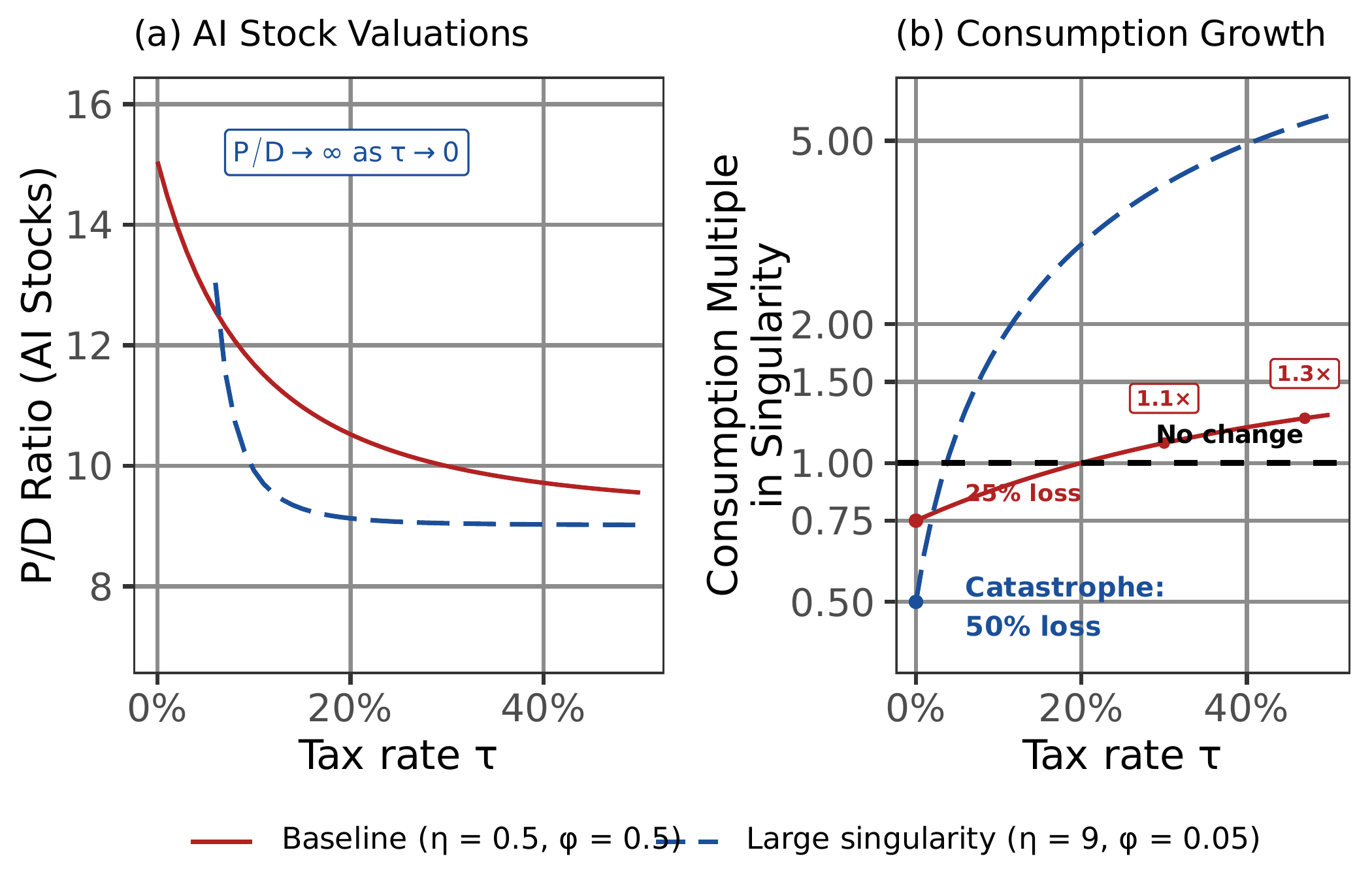}
\end{figure}

The policy implication is nuanced. In normal times, government transfers are a blunt and wasteful tool for addressing displacement risk. But if an AI singularity produces the kind of explosive output growth modeled by \citet{Jones2024} and examined critically by \citet{Nordhaus2021}, the calculus changes: the abundance of resources makes even inefficient redistribution effective. This suggests that contingent transfer policies---triggered by a singularity---may be worth designing in advance.

\section{Conclusion} \label{sec:conclusion} 

AI stocks trade at high valuations. We have argued that part of this premium reflects a hedging motive: investors use AI stocks to partially insure against an AI singularity that would displace their consumption. This mechanism requires market incompleteness---the inability to trade restricted AI equity---and is attenuated by extinction risk, which reduces the weight on the non-extinction states where the hedging channel operates. When markets are incomplete, risk-averse households may inefficiently block AI development, and government transfers, normally wasteful, can become effective if singularity-driven growth is large enough to overwhelm deadweight costs.

Our model is deliberately simple. It abstracts from continuous-time dynamics, heterogeneous beliefs, production-side details, and many other features that would enrich the analysis. The goal is not to provide a definitive account of AI stock valuations but to highlight a specific channel---hedging against displacement under incomplete markets---that connects asset pricing theory to the distinctive features of the AI singularity.

\appendix

\section{Proof of Proposition~\ref{prop:pd-ratios}} \label{app:proof-pd} 

The household's Euler equation for any asset $j$ with price $P_t^j$ and dividend $D_t^j$ is:
\begin{equation}
    P_t^j = \mathbb{E}_t \left[ \beta \left( \frac{c_{t+1}^H}{c_t^H} \right)^{-\gamma} (P_{t+1}^j + D_{t+1}^j) \right]. \label{eq:euler}
\end{equation}

Consider the AI stock. Write the price-dividend ratio as $v^{AI} = P^{AI}/D^{AI}$, which is constant in the stationary equilibrium before any singularity has occurred. The household's consumption growth $c_{t+1}^H / c_t^H$ takes three values:
\begin{itemize}
    \item No singularity (prob $1-p$): $c_{t+1}^H / c_t^H = 1+g$, and $D_{t+1}^{AI}/D_t^{AI} = 1+g$.
    \item Non-extinction singularity (prob $p(1-\xi)$): $c_{t+1}^H / c_t^H = \phi(1+g)(1+\eta)$, and $D_{t+1}^{AI}/D_t^{AI} = \Gamma^{AI} (1+g)$.
    \item Extinction (prob $p\xi$): $c_{t+1}^H = 0$, so the payoff is zero.
\end{itemize}

Substituting into the Euler equation with $P_t^{AI} = v^{AI} D_t^{AI}$:
\begin{equation} \label{eq:euler-expand}
    \begin{split}
        v^{AI} D_t^{AI} & = \beta(1+g)^{-\gamma} \Big[ (1-p)(1+g)(v^{AI}+1) D_t^{AI}                                \\
                        & \quad + p(1-\xi) [\phi(1+\eta)]^{-\gamma} (1+g) \, \Gamma^{AI} (v^{AI}+1) D_t^{AI} \Big].
    \end{split}
\end{equation}

Note that after a non-extinction singularity the economy resets with new parameters $(\alpha', \theta')$, and the \emph{new} P/D ratio in the post-singularity state will generally differ because $\Gamma^{AI}$ depends on $\theta$. (The non-AI growth factor $\Gamma^{N} = (1-\Delta\theta)(1+\eta)$ is $\theta$-independent, so the non-AI closed form is exact.) For tractability, we approximate the post-singularity AI P/D ratio by $v^{AI}$. This approximation becomes exact as $\Delta\theta \to 0$. For the numerically exact values reported in Table~\ref{tab:pd-ratios}, we solve the Euler equation by backward recursion over the chain $\theta_0, \theta_1 = \theta_0 + \Delta\theta(1-\theta_0), \theta_2, \ldots$, starting from a terminal state where $\theta$ is close to one and the approximation error vanishes. Dividing by $D_t^{AI}$ and solving for $v^{AI}$:
\begin{equation}
    v^{AI} = \frac{\beta(1+g)^{1-\gamma}\left[(1-p) + p(1-\xi)(1+\eta)^{-\gamma}\phi^{-\gamma} \Gamma^{AI}\right]}{1 - \beta(1+g)^{1-\gamma}\left[(1-p) + p(1-\xi)(1+\eta)^{-\gamma}\phi^{-\gamma} \Gamma^{AI}\right]}. \label{eq:pd-ai-solve}
\end{equation}

This can be rewritten as equation~\eqref{eq:pd-ai}. The derivation for non-AI stocks is identical, replacing $\Gamma^{AI}$ with $\Gamma^{N}$.

\clearpage
\section{The Loop's Work by Iteration (Human-Written)} \label{appendix:iter-log}

\providecommand{\rloopcommit}[2]{\href{https://github.com/chenandrewy/ralph-wiggum-asset-pricing/commit/#2}{\textcolor{blue}{\texttt{#1}}}}
\providecommand{\testresult}[1]{\href{https://github.com/chenandrewy/ralph-wiggum-asset-pricing/blob/42ce3a298a6c9e9ec144949346034b8af9055646/test-results/#1.md}{\textcolor{blue}{\texttt{#1}}}}
\providecommand{\penultimatetestresult}[1]{\href{https://github.com/chenandrewy/ralph-wiggum-asset-pricing/blob/36c1336aa3c27f41b14049fa5a8e1426214c76b9/test-results/#1.md}{\textcolor{blue}{\texttt{#1}}}}
\providecommand{\finaltestresult}[1]{\href{https://github.com/chenandrewy/ralph-wiggum-asset-pricing/blob/bf0af226ad2d15302dbe1c2878a1c2f66192e340/test-results/#1.md}{\textcolor{blue}{\texttt{#1}}}}

Full disclosure: this appendix is mostly human-written.

Before starting the loop, I used \href{https://github.com/chenandrewy/ralph-wiggum-asset-pricing/blob/main/check-ralph-direction.sh}{\textcolor{blue}{\texttt{check-ralph-direction.sh}}}, which runs only the author plan and improve steps of the loop, to generate \href{https://github.com/chenandrewy/ralph-wiggum-asset-pricing/tree/human-preface/ralph-garage/check-direction}{\textcolor{blue}{\texttt{several paper drafts}}}. I picked the most promising draft and initialized the loop with it.

The \href{https://github.com/chenandrewy/ralph-wiggum-asset-pricing/blob/main/ralph/ralph-loop.sh}{\textcolor{blue}{loop}} runs on its own branch, and commits after each iteration. So you can see in painstaking detail how the paper progressed on \href{https://github.com/chenandrewy/ralph-wiggum-asset-pricing/commits/ralph/run-final/}{\textcolor{blue}{\texttt{ralph/run-final}}}.

Usually, the very first \href{https://github.com/chenandrewy/ralph-wiggum-asset-pricing/blob/main/ralph/author-plan.py}{\textcolor{blue}{\texttt{author-plan}}} and \href{https://github.com/chenandrewy/ralph-wiggum-asset-pricing/blob/main/ralph/author-improve.py}{\textcolor{blue}{\texttt{author-improve}}} prompts generate a \href{https://github.com/chenandrewy/ralph-wiggum-asset-pricing/blob/human-preface/ralph-garage/history/001-paper.pdf}{\textcolor{blue}{fairly complete paper}}. Then, the first 3--4 iterations catch the major issues. This was evident in the final run:

\begin{itemize}
    \item \rloopcommit{iter 1}{5a7ff8dbb034cb39fa80637bf98d40c202c27fb8} fixed:
          \begin{itemize}
              \item Missing core literature on displacement risk and asset pricing, especially Kogan--Papanikolaou line of work.
              \item Opening figure uses fabricated data.
              \item Inconsistencies between paper's displayed P/D formulas and code.
              \item Abstract/body mismatch on extinction risk distorting AI development.
              \item Model inconsistency in treatment of private AI capital dividends and unused risk-free-rate notation.
              \item Abstract exceeding the 100-word limit.
              \item Missing clarification that paper does not model GKP-style entry of future innovators.
          \end{itemize}

    \item \rloopcommit{iter 2}{d82fa4ab58155f78336ca6634ff92529bdea5319} fixed:
          \begin{itemize}
              \item Overstatement of what GKP say about broader trading solving displacement risk.
              \item Introduction overusing AI-authorship device in a repetitive, awkward way.
              \item Verbal claims around valuation spreads and Proposition 3 not matching the paper's own results.
              \item Appendix/proof material being mostly ceremonial rather than doing real argumentative work.
          \end{itemize}

    \item \rloopcommit{iter 3}{3b369029a115df6ac8cadb918510cd21f06bc96b} fixed:
          \begin{itemize}
              \item Lit review exceeding half-page spec limit.
              \item Wrong attribution of JPE 2020 displacement/inequality paper and inconsistent terminology in transfer section.
              \item Long proof of Proposition 1 appearing inline instead of being moved to the appendix.
              \item Figure 2 text being too small to read.
              \item Colored hyperlink boxes showing up in compiled PDF.
          \end{itemize}

    \item \rloopcommit{iter 4}{42ce3a298a6c9e9ec144949346034b8af9055646} fixed:
          \begin{itemize}
              \item Opening figure still using fabricated data.
              \item Transfer figure not clearly showing that households are badly harmed with zero transfers.
              \item Citation errors around Kogan-Papanikolaou-related authorship and overly favorable framing of Nordhaus.
              \item Missing existence condition for Proposition 1, which made some Extension 2 valuations mathematically undefined.
          \end{itemize}
\end{itemize}

After the \texttt{iter 4} fixes, the remaining test failures were:

\begin{itemize}
    \item \testresult{element-gkp-cites}: Mischaracterization of GKP's mechanism, inaccurate/insensitive ``inessential extension'' language.
    \item \testresult{element-lit-review}: Missing Knesl from lit review.
    \item \testresult{factcheck-bysection}: Text says non-AI dividends ``shrink'' but they only shrink relative to the economy.
    \item \testresult{factcheck-freely}: Remaining bibliography problem around ``Left Behind''.
    \item \testresult{visual-figures-image-only}: Figure-formatting issues.
    \item \testresult{visual-pages}: Float / page-layout issue.
    \item \testresult{writing-intro}: Unneeded filler paragraphs, misplaced key material, tonal break caused by AI-authorship disclosure, main argument hard to follow on a skim.
\end{itemize}

These are minor, with the exception of \texttt{element-gkp-cites}. I totally would \textbf{not} put my name on a paper that mis-characterizes the key closely-related paper. \rloopcommit{iter 5}{57aae1aafc11092d83f50f3f87c1b2fb5abdccd2} fixed the problem.

The later iterations were focused on minor issues. For example:

\begin{itemize}
    \item \rloopcommit{iter 10}{668bb0e9dbe8ec1ac84ab514a849cb722be275a9}
          \begin{itemize}
              \item Attempted but failed: Tighten intro flow.
              \item Fixed: Figure 2 readability, Prop 2 (iii) stationarity discussion.
          \end{itemize}

    \item \rloopcommit{iter 11}{012e41ad24eea1f8175d8db39d9b9565a811ff1e}
          \begin{itemize}
              \item Attempted but failed: Fix extension figure whitespace.
              \item Fixed: Seru is not a co-author on Kogan-Papanikolaou-Stoffman 2020, negative dividend accounting problem, \texttt{writing-intro} flow and contribution clarity.
          \end{itemize}
\end{itemize}

The next 25 iterations were mostly failure churn. There would be 1-4 failures that would bounce around across tests, though \texttt{writing-intro} was a consistent blocker. The failures were rarely factchecks, so I feel one could have stopped the loop at any point after \texttt{iter 10} and just polished the rest by hand.

The final iterations were:

\begin{itemize}
    \item \rloopcommit{iter 35}{36c1336aa3c27f41b14049fa5a8e1426214c76b9}
          \begin{itemize}
              \item Remaining failure: \penultimatetestresult{visual-figures-image-only}, because the extension figure still had visual-formatting issues.
          \end{itemize}

    \item \rloopcommit{iter 36}{bf0af226ad2d15302dbe1c2878a1c2f66192e340}
          \begin{itemize}
              \item Remaining failures: \finaltestresult{spec-paper}, because Proposition 3's proof was still too long to remain inline; and \finaltestresult{writing-intro}, because the complete-markets counterfactual was still not visible enough to a skimming reader.
          \end{itemize}
\end{itemize}

At this point, I gave up. On my two \$200/month Claude Code subscriptions, I could only run about 6-8 iterations per day, and it was unclear these iterations were going to pass all tests.

In fact, I felt I had to intervene in the ``final'' paper draft, as the abstract had this bizarre sentence ``The displacement the paper models is closer than it appears.'' It also overclaimed in saying ``All analysis, code and prose were produced by AI agents.'' This last, final edit is documented in \href{https://github.com/chenandrewy/ralph-wiggum-asset-pricing/commit/4bed451b406d6c1190ef42e33517946c980d822d}{\textcolor{blue}{this commit}}.

\nocite{NovyMarxVelikov2025,LuEtAl2024,HadfieldMenellHadfield2018,Bengio2023,KorinekSuh2024,Zhang2019,BabinaEtAl2024,ChenLopezLiraZimmermann2022}
\printbibliography

@article{GKP2012,
  title={Displacement Risk and Asset Returns},
  author={G\^{a}rleanu, Nicolae and Kogan, Leonid and Panageas, Stavros},
  journal={Journal of Financial Economics},
  volume={105},
  number={3},
  pages={491--510},
  year={2012}
}

@article{Jones2024,
  title={The {AI} Dilemma: Growth versus Existential Risk},
  author={Jones, Charles I.},
  journal={American Economic Review: Insights},
  volume={6},
  number={4},
  pages={575--590},
  year={2024}
}

@article{Barro2006,
  title={Rare Disasters and Asset Markets in the Twentieth Century},
  author={Barro, Robert J.},
  journal={Quarterly Journal of Economics},
  volume={121},
  number={3},
  pages={823--866},
  year={2006}
}

@article{Wachter2013,
  title={Can Time-Varying Risk of Rare Disasters Explain Aggregate Stock Market Volatility?},
  author={Wachter, Jessica A.},
  journal={Journal of Finance},
  volume={68},
  number={3},
  pages={987--1035},
  year={2013}
}

@article{PastorVeronesi2009,
  title={Technological Revolutions and Stock Prices},
  author={P\'{a}stor, \v{L}ubo\v{s} and Veronesi, Pietro},
  journal={American Economic Review},
  volume={99},
  number={4},
  pages={1451--1483},
  year={2009}
}

@article{Acemoglu2025,
  title={The Simple Macroeconomics of {AI}},
  author={Acemoglu, Daron},
  journal={Economic Policy},
  volume={40},
  number={121},
  pages={13--58},
  year={2025}
}

@article{Nordhaus2021,
  title={Are We Approaching an Economic Singularity? {I}nformation Technology and the Future of Economic Growth},
  author={Nordhaus, William D.},
  journal={American Economic Journal: Macroeconomics},
  volume={13},
  number={1},
  pages={299--332},
  year={2021}
}

@article{KoganPapanikolaou2014,
  title={Growth Opportunities, Technology Shocks, and Asset Prices},
  author={Kogan, Leonid and Papanikolaou, Dimitris},
  journal={Journal of Finance},
  volume={69},
  number={2},
  pages={675--718},
  year={2014}
}

@article{KoganPapanikolaouStoffman2020,
  title={Left Behind: Creative Destruction, Inequality, and the Stock Market},
  author={Kogan, Leonid and Papanikolaou, Dimitris and Stoffman, Noah},
  journal={Journal of Political Economy},
  volume={128},
  number={3},
  pages={855--906},
  year={2020}
}

@article{Knesl2023,
  title={Automation and the Displacement of Labor by Capital: Asset Pricing Theory and Empirical Evidence},
  author={Knesl, Ji\v{r}{\'i}},
  journal={Journal of Financial Economics},
  volume={147},
  number={2},
  pages={271--296},
  year={2023}
}

@incollection{AghionJonesJones2019,
  title={Artificial Intelligence and Economic Growth},
  author={Aghion, Philippe and Jones, Benjamin F. and Jones, Charles I.},
  booktitle={The Economics of Artificial Intelligence: An Agenda},
  publisher={University of Chicago Press},
  pages={237--290},
  year={2019}
}

@article{NovyMarxVelikov2025,
  title={Artificial Intelligence-Powered (Finance) Scholarship},
  author={Novy-Marx, Robert and Velikov, Mihail},
  journal={Journal of Economic Literature},
  volume={64},
  number={1},
  pages={5--37},
  year={2026},
  doi={10.1257/jel.20251821}
}

@misc{LuEtAl2024,
  title={The {AI} Scientist: Towards Fully Automated Open-Ended Scientific Discovery},
  author={Lu, Chris and Lu, Cong and Lange, Robert Tjarko and Foerster, Jakob N. and Clune, Jeff and Ha, David},
  year={2024},
  eprint={2408.06292},
  archivePrefix={arXiv},
  primaryClass={cs.AI}
}

@misc{HadfieldMenellHadfield2018,
  title={Incomplete Contracting and {AI} Alignment},
  author={Hadfield-Menell, Dylan and Hadfield, Gillian K.},
  year={2018},
  eprint={1804.04268},
  archivePrefix={arXiv},
  primaryClass={cs.AI}
}

@article{Bengio2023,
  title={{AI} and Catastrophic Risk},
  author={Bengio, Yoshua},
  journal={Journal of Democracy},
  volume={34},
  number={4},
  pages={111--121},
  year={2023}
}

@misc{KorinekSuh2024,
  title={Scenarios for the Transition to {AGI}},
  author={Korinek, Anton and Suh, Donghyun},
  year={2024},
  eprint={2403.12107},
  archivePrefix={arXiv},
  primaryClass={econ.GN}
}

@article{Zhang2019,
  title={Labor-Technology Substitution: Implications for Asset Pricing},
  author={Zhang, Miao Ben},
  journal={Journal of Finance},
  volume={74},
  number={4},
  pages={1793--1839},
  year={2019},
  doi={10.1111/jofi.12766}
}

@misc{BabinaEtAl2024,
  title={Artificial Intelligence and Firms' Systematic Risk},
  author={Babina, Tania and Fedyk, Anastassia and He, Alex Xi and Hodson, James},
  year={2024},
  note={Available at SSRN 4868770},
  doi={10.2139/ssrn.4868770}
}

@misc{ChenLopezLiraZimmermann2022,
  title={Does Peer-Reviewed Research Help Predict Stock Returns?},
  author={Chen, Andrew Y. and Lopez-Lira, Alejandro and Zimmermann, Tom},
  year={2022},
  eprint={2212.10317},
  archivePrefix={arXiv},
  primaryClass={q-fin.PM}
}

@misc{AnthropicInstitute2026,
  author={{Anthropic Institute}},
  title={When {AI} Builds Itself: Our Progress toward Recursive Self-Improvement, and Its Implications},
  year={2026},
  howpublished={\url{https://www.anthropic.com/institute/recursive-self-improvement}},
  urldate={2026-07-29}
}

@misc{Wu2026,
  author={Wu, Cat},
  title={How {Anthropic's} Product Team Moves Faster than Anyone Else},
  year={2026},
  howpublished={Interview by Lenny Rachitsky, \emph{Lenny's Podcast}, \url{https://www.lennysnewsletter.com/p/how-anthropics-product-team-moves}},
  urldate={2026-07-29}
}

@misc{DavisEtAl2026,
  title={Cheap Code, Costly Judgment: A Case Study on Governable Agentic Software Engineering},
  author={Davis, James C. and Amusuo, Paschal C. and Singla, Tanmay and {\c{C}}akar, Berk and Davis, Kirsten A.},
  year={2026},
  eprint={2607.01087},
  archivePrefix={arXiv},
  primaryClass={cs.SE}
}

@online{Huntley2025,
  author={Huntley, Geoff},
  title={Ralph Wiggum Loop},
  year={2025},
  url={https://ghuntley.com/ralph/},
  urldate={2026-07-29}
}

@techreport{AndrewsFarboodi2025,
  author={Andrews, Isaiah and Farboodi, Maryam},
  title={Do Markets Believe in Transformative {AI}?},
  institution={National Bureau of Economic Research},
  type={NBER Working Paper},
  number={34243},
  year={2025},
  doi={10.3386/w34243},
  url={https://www.nber.org/papers/w34243},
  urldate={2026-07-29}
}

\end{document}